\documentclass[12pt]{amsart}
\usepackage{amssymb}
\usepackage{amsmath}
\usepackage{amsfonts}
\usepackage{mathrsfs}
\usepackage{graphicx}
\usepackage{color}
\usepackage[onehalfspacing]{setspace}
\usepackage{caption}
\usepackage{enumerate}
\usepackage[round]{natbib}
\usepackage{appendix}
\usepackage{lscape}
\usepackage{subcaption}
\usepackage{graphicx}
\usepackage{amsfonts}
\usepackage{placeins}
\usepackage[utf8]{inputenc}
\usepackage{charter}
\usepackage[colorlinks=true,citecolor=black,urlcolor=blue,pdfpagemode=UseNone,pdfstartview=FitH]{hyperref}
\usepackage{apptools}
\usepackage{multibib}
\usepackage{multirow}
\newcites{main,supp}{References,References}
\makeatletter
\def\section{\@startsection{section}{1}
	\z@{1.0\linespacing\@plus\linespacing}{.8\linespacing}{\Large}}

\def\subsection{\@startsection{subsection}{2}
	\z@{.8\linespacing\@plus.7\linespacing}{.7\linespacing}{\large}}

\def\subsubsection{\@startsection{subsubsection}{3}
	\z@{.5\linespacing\@plus.7\linespacing}{-.5em}{\normalfont\bfseries}}
\makeatother

\numberwithin{equation}{section}

\newtheorem{claim}{Claim}[]

\theoremstyle{definition}

\theoremstyle{definition}

\theoremstyle{definition}

\title{}
\begin{document}
	\vspace*{3ex minus 1ex}
	\begin{center}
		\Large \textsc{Choosing The Best Incentives for Belief Elicitation with an Application to Political Protests} 
		\bigskip
	\end{center}
	
	\date{%
		\today%
	}

\date{\today}
\vspace*{3ex minus 1ex}
	\begin{center}
		Nathan Canen and Anujit Chakraborty\\
	
	\bigskip
	
	
	\end{center}
	
	\thanks{We thank Nina Bobkova, Juan Felipe Ria\~no, Sebastian Saiegh, Ko Sugiura and Francesco Trebbi for very helpful comments and suggestions. All errors are our own.\\
	\textbf{Canen:} (Corresponding Author) University of Houston, TX, USA and Research Economist at NBER. email: \url{ncanen@uh.edu} \\
	\textbf{Chakraborty:} University of California - Davis, CA, USA. email: \url{chakraborty@ucdavis.edu}}
	
	\begin{abstract}
Many experiments elicit subjects' prior and posterior beliefs about a random variable to assess how information  affects one's own actions. However, beliefs are multi-dimensional objects, and experimenters often only elicit a single response from subjects. In this paper, we discuss how the incentives offered by experimenters map subjects' true belief distributions to what profit-maximizing subjects respond in the elicitation task. In particular, we show how slightly different incentives may induce subjects to report the mean, mode, or median of their belief distribution. If beliefs are not symmetric and unimodal, then using an elicitation scheme that is mismatched with the research question may affect both the magnitude and the sign of identified effects, or may even make identification impossible. As an example, we revisit \cite{cyyz}'s study of whether political protests are strategic complements or substitutes. We show that they elicit modal beliefs, while modal and mean beliefs may be updated in opposite directions following their experiment. Hence, the sign of their effects may change, allowing an alternative interpretation of their results.
\end{abstract}
	\maketitle
	
\noindent JEL: C81, C93, D74, P00\\
\noindent Keywords: Belief Elicitation, Experimental Designs, Identification, Political Protests
		
	\section{Introduction}
	
Economic theory posits that an individual's actions in a strategic environment depend on their preferences over outcomes and their beliefs about how others would act. 
However, measuring the causal role of beliefs on actions from observational data is inherently difficult. After all, even when actions are observable, beliefs are rarely observed. Meanwhile, imputing beliefs from actions requires presupposing their causal effect, even though that is often the empirical goal itself. 

Experimental methods provide an attractive solution in these contexts. In experiments, researchers can observe actions \textit{and} elicit beliefs, both before and after carefully designed information interventions. For example, within the context of political activism, \citet{cyyz, jarke2021free, hager2022group, hager2022political} all elicit subjects' beliefs about others' planned participation in political events/protests. They then measure how changing subjects' information about others' participation causally affects individual beliefs and, thence, individuals' own attendance.\footnote{Such examples go beyond political economy. In another recent example, \citet{bursztyn} evaluate whether Saudi husbands are more likely to support women working outside the home if they discover that other husbands do so too.}

However, eliciting beliefs is fundamentally different than eliciting simpler variables, such as willingness-to-pay. This is because beliefs are a probability distribution, often defined over a large set of outcomes, rather than just a point response (i.e., a real number between 0 and 1). For example, the belief about the proportion of $N$ other survey-participants who participate in a protest is a probability distribution over the $N+1$ possible values $\{0,1/N, 2/N,..1\}$. Yet, for tractability, researchers often have to elicit point-responses in their experimental design, and interpret those as a coarse measure of beliefs. This is the case in the papers cited above, as well as many others.\footnote{\cite{trebbi, trebbi2} are some exceptions in the experimental political economy context.} 

In this paper, we discuss \emph{the relation between subjects' point-responses and their underlying belief distribution}, and how this mapping depends on incentives offered to the subject. In particular, we compare two popular belief elicitation schemes that seem superficially similar, but, as we show in Section \ref{sec: belief}, incentivize different best responses in belief reporting. We then show the empirical consequences of such differences, which can include identifying effects with opposite signs to the true ones, or a lack of identification altogether. 

The first scheme we consider rewards subjects for correct guesses within an error band of the true value: for instance, ``Please guess $x\in [0,1]$. If your guess is within $\Delta$ percentage points of $x\in [0,1]$ you will earn a bonus payment of 1 currency unit." Recent examples of belief elicitation using this scheme include \cite{cyyz, chenyang, bursztyn}, among others. In Section \ref{sec:mode}, we prove that subjects' best response to these incentives is to report the (approximate) mode of the true distribution of $x$ (henceforth, modal beliefs). The second one, which we recommend to practitioners who wish to elicit the mean of the belief distribution over $x$ (mean beliefs, henceforth), rewards subjects $A-B(x-r)^2$ for reporting $r$, where $A,B$ are constants.  
Such an incentive scheme indeed induces profit-maximizing subjects to report their mean beliefs as a best response. 

The difference in elicitation designs is subtle: rewards under both schemes are weakly increasing in accuracy. But they induce very different mappings from the true belief distribution (over $x$) to the optimal report, as mean and mode do not generally coincide. For the sake of completeness, we also discuss a third incentive scheme, which rewards subjects based on the absolute distance between their report and the true value. This incentivizes reporting the median of the true distribution instead.

We show that these differences can be very consequential when it comes to testing a theoretical prediction. Not only are modes, means and medians generally different, but they may have very different properties following an information intervention. This can be visualized within a simple example, shown in Figure \ref{fig1}. In this example, beliefs about a random variable, $x$, are assumed to follow a flexible continuous distribution, which has been calibrated to our empirical application (described below). In Figure \ref{fig1}, the mode is given by 0.142, while the mean is 0.273. Now, we assume there is an experimental information intervention, like those in the papers cited above. In particular, the treatment reveals that $x$ is equal to 0.17, akin to \cite{cyyz}. Beliefs will then update towards 0.17, implying the mode will increase towards 0.17, while the mean will \textit{decrease} accordingly, as illustrated by the arrows in the figure.

\begin{figure}
	\centering
		\includegraphics[width=\textwidth]{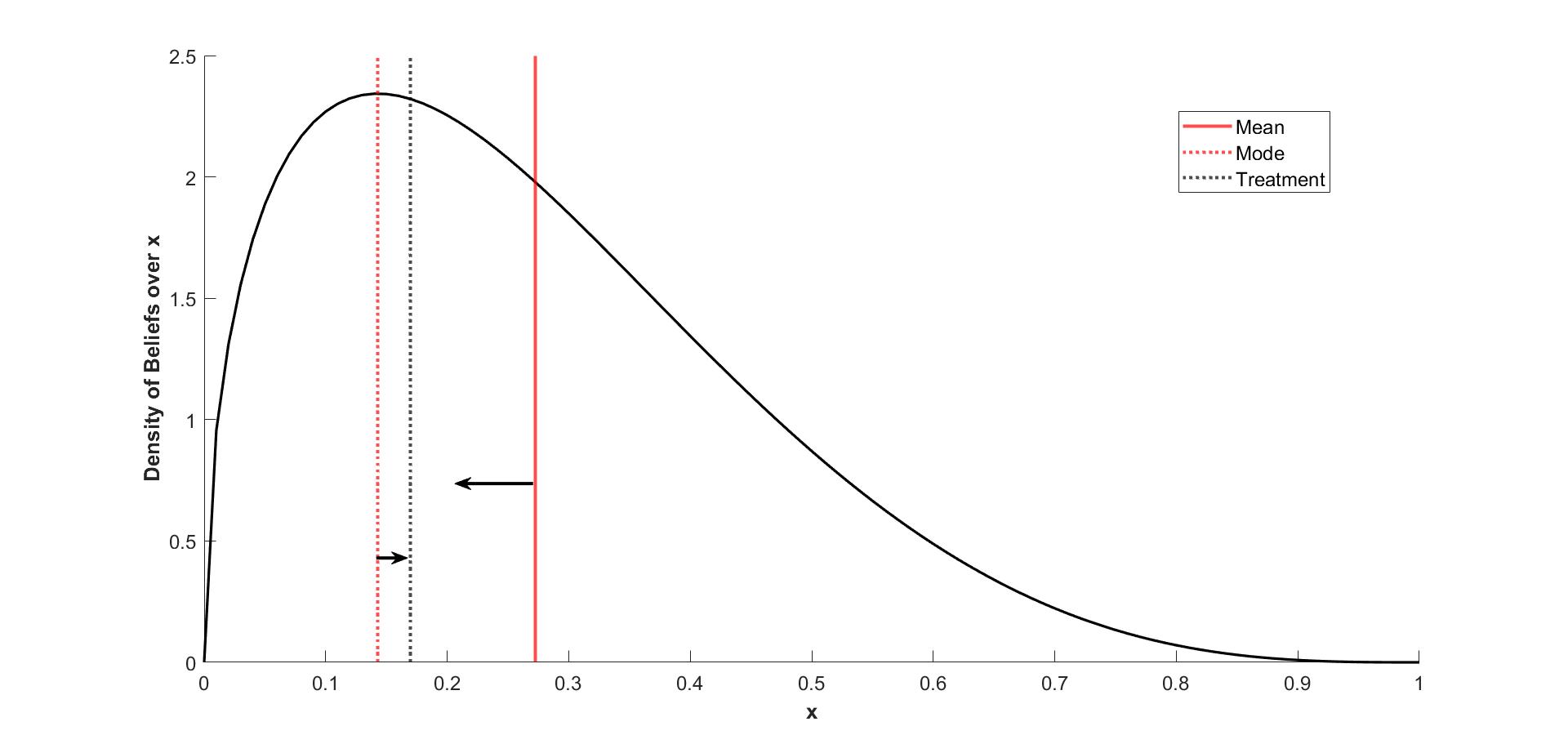}
	\caption{Simple Illustration that Means and Modes are Updated in Opposite Directions}
		\label{fig1}
		\parbox{6.2in}{\footnotesize We present an example of how means and modes can be updated in opposite directions, implying that different belief elicitation schemes may provide different signs for identified effects. In the example, beliefs follow $x\sim Beta(1.5,4)$, where the parameters are calibrated to the empirical application discussed in Section \ref{subsec: quantification}. The mode (dotted red line) is equal to 0.142, while the mean (solid red) equals 0.273. An information intervention (black dotted line) that reveals the true state to be 0.17 is introduced to subjects. Beliefs are updated towards 0.17, which implies that the mode increases, but the mean decreases (shown by the black arrows).}	
\end{figure}

Hence, a researcher who uses incentives to measure modal beliefs but interprets the data as mean beliefs, will have identified a very different measure than they intended to. And using this measure may lead to very different conclusions while testing a theory-derived hypothesis. In fact, as the example above shows, researchers who interpret modal and mean beliefs interchangeably, may (i) fail to identify the theoretical parameter of interest, or may (ii) identify their \textit{opposite sign} (if the parameter is defined by how changes in mean beliefs relate to changes in actions). We formalize these points in Section \ref{why_it_matters}.

As an empirical application, we revisit \cite{cyyz}'s experiment on political protests in Hong Kong.\footnote{We also discuss \cite{hager2022group, jarke2021free, hager2022political} that use a similar research design and where our insights also apply.} \cite{cyyz} conduct a cleverly designed field experiment in the backdrop of an anti-authoritarian protest in Hong Kong. They show that potential protesters who increased their reported beliefs in favor of others being more likely to protest, became less likely to protest themselves. Interpreting the reported beliefs as mean beliefs, \cite{cyyz} conclude that ``specifically, our findings reject many recent models that assume only the possibility of strategic complementarity in the protest decision." (p.1072) That is, they interpret their evidence as contradicting the widespread assumption that political protests are strategic complements (i.e., when a citizen believes that more peers are likely to protest, she should become more likely to protest herself, see \cite{gehlbach}). 

However, we argue that the empirical results in \cite{cyyz} do not necessarily reject strategic complementarity (SC) or the many theoretical models that assume it. Rather, we show that \cite{cyyz}'s evidence can be completely consistent with SC using a two-fold argument. First, \cite{cyyz}'s experimental design uses the first belief-elicitation scheme described above, thereby incentivizing subjects to report their modal beliefs, but not their mean beliefs. Thus, their belief data does not necessarily reveal mean beliefs. Second, we use a stylized theoretical model and closely derived statistical results to show that means and modes can be updated in \textit{opposite directions} in their set-up, as shown in Figure \ref{fig1}. Therefore, their estimates of a negative covariance between belief updating about others' protest attendance and one's own attendance may be flipped.

The distinction borne in our discussion is also empirically meaningful. We perform a numerical exercise where we fit flexible distributions of individual beliefs to the distribution of observed prior and posterior reports in the \cite{cyyz} dataset, assuming subjects were reporting modal beliefs about others' participation. Based on the estimated distribution of beliefs, we find that over a third of subjects would have updated their mean belief in the opposite direction of their reported belief. Therefore, how one interprets ``elicited beliefs'' determines what the data of \cite{cyyz} implies about strategic complementarity.

To be clear, our analysis does not aim to refute or reverse any regression results from \cite{cyyz}. Rather, we use it as a prominent example of how elicitation schemes influence the interpretation of economic data: The scheme used in \cite{cyyz} allows alternative interpretations of their results, and such interpretations are meaningful to what we can learn about the empirical validity of strategic complementarity. Such alternative explanations could have been completely avoided with the alternative mean-eliciting scheme that we propose. None of what we say precludes the possibility that subjects have unimodal symmetric beliefs and, thus, their mean and modal beliefs are identical, thereby preserving their original interpretation in \cite{cyyz}. Instead, we think of our results as clarifying the important role that this implicit assumption (unimodal symmetric beliefs) plays in linking the regression results to the final conclusion, both in their experiment and in others, emphasizing the importance of the design of belief elicitation procedures per se.

\section{Belief Elicitation}\label{sec: belief}

Suppose subjects believe that some outcome $x\in[0,1]$ is distributed according to $f$, a probability density function over $[0,1]$. In \cite{cyyz} and \cite{hager2022group}, for instance, $x$ is the percentage of the participants from the study who plan to participate in a political protest. In \cite{bursztyn}, it is the percentage of Saudi husbands who believe women should be allowed to work outside the home. Suppose researchers are interested in eliciting a single value that is in some way representative of the distribution $f$.

\subsection{Elicitation Schemes that Incentivize Reporting the Mode}\label{sec:mode}

Consider the following general incentive system:

\medskip

\begin{quote}
``Please guess $x\in [0,1]$. If your guess is within $\Delta$ percentage points of $x\in [0,1]$ you will earn a bonus payment of 1 currency unit."
\end{quote} 

\medskip

\cite{cyyz} use a special case of this with $\Delta=2$, and where 1 currency unit was 10 HKD. Similarly, \cite{chenyang} sets $\Delta = 0.1$ and the currency unit was RMB 5. Meanwhile \cite{bursztyn} uses some $\Delta$ close to 0 with strategic uncertainty about its value and US\$20 as a currency unit.\footnote{Technically, the incentives in \cite{bursztyn} award the participant ``who guesses most accurately". This could induce strategic guessing, where a subject anticipates other subjects' guesses when choosing their report. However, with a large enough population of independent subjects (1500 in their sample), payment should reflect being arbitrarily close to the true values ($\Delta \to 0$).}

There exists simple sufficient conditions on $f$ under which, under this elicitation scheme, profit maximizing subjects should report the mode exactly, or approximately with an error band of $\Delta$ percentage points. If $\Delta = 0$ then the mode will trivially be the best response. The interesting case is when $\Delta>0$.

\begin{claim}
\label{claim: mode}
Let $f$ be unimodal and $m=\arg\max_x f(x)$ be its mode. Let $\delta=\Delta/100$.

i) If $f$ is discrete and $\min_{x \in [0,1], y \in [0,1]}\{|x-y|: f(x)>0,f(y)>0\}>\delta$, then the profit-maximizing report is $m$.

ii) If $f$ is differentiable, $f'(x)\leq0$ for all $x\in[0,m]$ and $f'(x)\geq0$ for all $x\in[m,1]$, then the profit-maximizing report lies within $[m-\delta,m+\delta]$ for all $\delta$.

iii) If $f$ is differentiable and there exists an $\epsilon >\delta $ such that $f'(x)\leq0$ for all $x\in[m-\epsilon,m]$, $f'(x)\geq0$ for all $x\in[m,m+\epsilon]$, and $\min_{x\in[m-\epsilon,m+\epsilon]}f(x)\geq\max_{x\notin[m-\epsilon,m+\epsilon]}f(x)$, then the profit-maximizing report lies within $[m-\delta,m+\delta]$.
\end{claim}

\begin{proof}
See Appendix \ref{appA}.
\end{proof}

Part (i) provides a very simple sufficient condition under which subjects would exactly report the mode when beliefs are discrete. We will use (i) in a stylized model in Subsection \ref{toy} to show that, following the information-treatment, the mode of $P_{-i}$ can be updated in the \textit{opposite} direction to the mean of $P_{-i}$. 

Part (ii) then characterizes the optimal report for the simple case of a continuous distribution where the density function has exactly one peak at the mode. In this case, the subject will (approximately) report the mode. Many density functions are single-peaked. For example, the $Beta$ distribution $f(x;\alpha,\beta)$ is single-peaked for any $\alpha,\beta>1$, a property we show formally in Claim \ref{claim:beta} in Appendix and use in Subsection \ref{subsec: quantification}. Part (iii) proves the same result for the more general case where the density function is only locally single-peaked around the mode.

As a result, the papers cited above which elicit subjects using scheme \ref{claim: mode} will obtain reports of subjects' modal beliefs about others' behavior (e.g., protest attendance, opinions on women's work rights, etc.), rather than their mean beliefs. 

\subsection{Elicitation Schemes that Incentivize Reporting the Mean}
\label{sec:alternative}

We now consider an alternative scheme that would elicit the mean instead, before discussing why these differences matter for empirical research.

Suppose subjects who reported $r$ are paid
$B-A(x-r)^{2}$ if $x \in [0,1]$ is the true/realized value from $f$. $B$ and $A$ are constants chosen by the experimenter with $B>A>0$. 

\begin{claim}
\label{claim: mean}
Let $x \sim f(\cdot)$ and $\bar{x}=\int_{0}^{1}xf(x)dx$. Then $\bar{x}$ is the profit-maximizing report under this elicitation scheme.
\end{claim}
\begin{proof}
In Appendix \ref{appA}.
\end{proof}

The proof is based on the mean being the minimizer of Mean Squared Error. Though the proof treats $f(\cdot)$ as a continuous distribution, a very similar proof would work even if $f(\cdot)$ were discrete. 

Notice that both incentive schemes presented so far are based on the same concept: they reward subjects depending on how close their guesses are to the true realization of the random variable, $x$. However, the details of those incentives matter: whether the reward scheme is based on a quadratic loss function, or on a near-enough value can elicit completely different reports. Beyond unimodal symmetric distributions, modes and means do not coincide.
This can affect both the magnitude and the sign of the conclusions drawn from such experiments, as we explain in Section \ref{why_it_matters}. Before that, though, we comment on two last, but relevant, elicitation schemes.

\subsection{Elicitation Schemes that Incentivize Reporting the Median}
\label{sec: median}

To reinforce our point about how incentive schemes matter, consider a slightly different incentive scheme where subjects who reported $r$ are paid
$B-A|x-r|$ if $x \in [0,1]$ is the true proportion of participants' outcomes, and $B$ and $A$ are constants chosen by the
experimenter with $B>A>0$. Even this small change to the reward function (from square to absolute value) ends up eliciting a different point of the distribution, in this case, the median. We show this separately for discrete and continuous distributions as a discrete distribution allows multiple medians. 

\begin{claim}
\label{claim: median}
i) If $f$ is differentiable and if $\int_{0}^{M} f(s)ds=\int_{M}^{1} f(s)ds$, then $M$ is the profit-maximizing report.

ii) If $f$ is discrete and for any $M$ such that $\sum_{s_i \leq M} f(s_i) = \sum_{s_i \geq M} f(s_i)$, $M$ is a profit-maximizing report.

\end{claim}
\begin{proof}
In Appendix \ref{appA}.
\end{proof}

\subsection{Experimental Designs Without Incentivized Reporting}\label{sec:noinc}

Certain experimental designs do not provide any incentives for subjects when they report their beliefs about others' protest attendance. For example, \cite{hager2022group} simply ask ``How many protesters do you think will attend the AfD protest?" and ``how many people will participate in the respective protests?", without rewards for correct guesses. Similar questions are found in \cite{jarke2021free, hager2022political}, who do not reward subjects for correct responses either. The authors explain that incentives could ``undermine this code [of honor for truthful answers] and make non-truthfulness salient in the first place" (\cite{jarke2021free}) and that ``the campaign were concerned that incentives would be perceived as very unusual by their supporters." (\cite{hager2022political}). 

While such reservations from institutional partners are natural, and while we understand that economists are forced to accept such constraints when working in the field, we also want to draw attention to the following point: Without incentives, it is no longer possible to map a subject's point-reports to their belief distribution. Recall our setting where subjects believe that some outcome $x\in[0,1]$ is distributed according to $f$. Absent any rewards, subjects' best responses to the elicitation questions above are not well-defined: Even rational subjects could be reporting the mode (as in \ref{sec:mode}), the mean (as in \ref{sec:alternative}), the median (as in \ref{sec: median}), or possibly, any other values within their support of beliefs. Perhaps, if beliefs were unimodal and symmetric, then subjects might find it "natural" to report the mean. Yet, this again requires using a strong assumption on an unobserved distribution. 

\section{Mean vs Mode: Why the Belief Elicitation Scheme Matters}\label{why_it_matters}

\subsection{Unobservability of Mean Beliefs in Theory-Driven Specifications}

		
To link the belief elicitation schemes above to actions/outcomes, consider an environment where an individual $i$ has a binary decision (e.g., to attend a protest or not, or whether to turn out to vote). This binary decision is denoted $P_i \in \{0,1\}$. Denote $P_{-i}$ as the proportion of other individuals choosing $P_i =1$, and the expected belief of $i$ over the latter as $\mathbb{E}_i P_{-i}$. For simplicity, let us write this best response as a linear specification, which is consistent with most of the empirical specifications cited above\footnote{It is clear that this linearity is only for exposition and to be closer to the linear and additive empirical specifications, as it can be derived from fundamentals of the game environment, such as linear-quadratic preferences.}:
		\begin{equation} \label{eq:participation}
	P_i = \delta_1 + \delta_2 \mathbb{E}_i P_{-i} + \varepsilon_i
		\end{equation}
The most common informational interventions cited in the previous sections manipulate the belief distribution over $P_{-i}$ to then manipulate $\mathbb{E}_i P_{-i}$ and identify $\delta_2$. For instance, \cite{bursztyn} reveals the share of Saudi husbands who reported that their wives should work outside the home, while \cite{cyyz} report the share of Hong Kong students who planned to attend a protest. The target parameter is $\delta_2$, which measures the extent to which individuals react to information about others' behavior. In particular, if $\delta_2>0$, it is a game of strategic complements, while if $\delta_2<0$, it is one of strategic substitutes. 

The participation decision in equation \ref{eq:participation} is based on the \textit{mean of $i$'s belief distribution}, which is consistent with expected-utility maximizing agents. 
However, if the elicitation task from \ref{sec:mode} is used, then subjects would report the \textit{mode of their belief distribution over $P_{-i}$}, either exactly or approximately, and not the mean of their beliefs over $P_{-i}$. Thus, $\mathbb{E}_i P_{-i}$ is not generally observed in the subjects' responses to the researchers' questions. Hence, unless the belief distribution has an identical mean and mode, $\delta_2$ is not identified.

\subsection{Means and Modes Can Be Updated in Opposite Directions}	

Now, suppose that a researcher elicits modal beliefs through an elicitation scheme in Section \ref{sec:mode}, but interprets it as though it was the mean belief with the goal of identifying $\delta_2$. It turns out that, if the mean and mode are actually not identical, then the researcher may not only fail to identify the magnitude of their desired effect, but may fail to identify the sign of the effect of interest.

To see this visually, we return to Figure \ref{fig1}. 
The information intervention that reveals $x=0.17$ increases modal beliefs about $x$, while mean beliefs about $x$ decrease.
Claim \ref{claim: mean-mode} in Section \ref{subsec: quantification} formalizes sufficient conditions for this class of distributions upon which, if subjects are reporting modal beliefs, while the econometrician interpreted it as mean beliefs, the identified sign of $\delta_2$ in equation \ref{eq:participation} would be the opposite of the true value.

\section{Political Protests as Strategic Complements: An Empirical Application}\label{empirical}

In this section, we revisit the empirical evidence in \cite{cyyz} in light of the belief elicitation schemes discussed in Sections \ref{sec:mode} and \ref{sec:alternative}. 

\cite{cyyz} motivate their experimental design as being able to identify whether political protests are strategic complements (i.e., when a citizen believes more agents are likely to protest, she should become more likely to protest herself) or strategic substitutes (e.g., when citizens have an incentive to free-ride on the costly participation of others, as in \citealp{tullock1971,palfrey1984,olsonlogic}). This maps into whether $\delta_2$ in equation (\ref{eq:participation}) is positive or negative.

The strategic complementarity assumption has been widely used in recent models (e.g., \citealp{mesquita,passarelli,barbera}), leading \cite{gehlbach} to write that strategic complementarity (SC, henceforth) ``characterizes mass protest" (p.579).\footnote{It is also intuitive, as it can be rationalized if citizens enjoy an increased benefit from attending a larger protest through an increased likelihood of a regime change, or through increased social interactions from such behavior, or through a lower cost (one is less likely to be arrested amongst a larger crowd if there is a crackdown by the incumbent regime).} However, whether protests are strategic complements or substitutes has important welfare effects: political protests are strongly associated with regime changes (see \cite{chenoweth}), but policy implications would differ whether protests are strategic substitutes or complements. 

Within this context, \cite{cyyz} conduct an important field experiment among students in Hong Kong. They find that potential protesters who were presented with the information that others were more likely to protest than their previously reported beliefs (elicited using the scheme in Section \ref{sec:mode}), became less likely to protest themselves. \cite{cyyz} conclude that this is evidence against models that solely assume that protests are strategic complements. (p.1072)

We now show that \cite{cyyz}'s evidence can actually be completely consistent with SC, thereby providing an alternative interpretation of their results. Our arguments are based on: (i) that their belief elicitation scheme satisfies Claim \ref{claim: mode}, thereby eliciting modal beliefs, (ii) \cite{cyyz} still interpret reported beliefs as mean beliefs (using an analogue to equation \ref{eq:participation}), even though modes and means can be updated in opposite direction. Hence, they may identify an effect with opposite sign of their desired one. 

We first overview their experiment, before illustrating these points using a stylized model and numerical exercises.

\subsection{The Original Experiment}\label{original_exp}

The experiment in \cite{cyyz} consists of three parts whose relevant features are overviewed here. Further details can be found in Appendix \ref{appB}.

In Part 1 of the experiment, one week before the planned protest, subjects were asked:
\begin{quotation}
(i) \textit{``Are you planning to participate in the July 1st march in 2016?"}\\

\noindent (ii) \textit{``Please guess what percentage of the participants from HKUST of this study plan to participate in the July 1st march in 2016...If your guess is within 2 percentage points of the percent of students who actually answer either ``Yes" or ``Not sure yet, but more likely than not," you will earn a bonus payment of HK\$10."}
\end{quotation}

Question (i) elicits the ex-ante decision for each subject, and the authors find that 17\% of subjects plan to attend the protest. Meanwhile, question (ii) elicits their (reported) prior belief about others' participation. As alluded to previously, it is a special case of the elicitation scheme in Claim \ref{claim: mode}.

In Part 2 of the experiment (one day before the protest), all subjects were reminded of their Part 1 answers to (i) and (ii) and allowed to update the responses. However, a random subset of subjects (i.e, the ``Treated" group) was given the information that 17\% of subjects in the previous survey were planning to attend the protest. Part 3 of the experiment was conducted after the protest, and asked subjects to self-report whether they turned out ($P_i = 1$) or not ($P_i = 0$).

The authors estimate an analogue to equation \ref{eq:participation}. Their target parameter is $\delta_2$, as they are in interested in how changes to one's beliefs affect participation decisions.
 
\subsection{A Stylized Model: General Description}\label{toy}

In Appendix \ref{subsec:counter}, we provide a stylized model where we work through what an experimentalist would observe in terms of elicited beliefs and participation decisions, if subjects maximized their expected utility in \cite{cyyz}'s set-up. Here, we give a general outline of the model's set-up and insights.

Subjects in the model make two choices. They attend the protest ($P_i = 1$), or not ($P_i = 0$), and they choose some value to report, $r_i$, about others' participation, given their beliefs and the researchers' incentives described in Section \ref{original_exp}. As a rhetorical device, we construct subjects' preferences to satisfy SC: i.e., their preferences to attend are increasing in others' participation levels, $P_{-i}$. There is a cost of $0.165$ to attend the protest, so that $i$ wants to attend if \textit{the mean of others' attendance} is larger than $16.5\%$.  Critically, given the incentives used for belief-elicitation, even though the actual decision to protest depends on mean beliefs, the decision on what to report depends on modal beliefs.

Subjects can belong to two groups, which only differ in their distribution of beliefs over $P_{-i}$. Both groups have prior mean beliefs below $0.165$, so neither wants to attend the protest ex-ante. However, Group 1 has a high modal belief over $P_{-i}$, while Group 2 has low modal beliefs. Following the global games literature (see \cite{morris} for an overview), we model the information intervention (Part 2 of the experiment) as a noisy signal which has a distribution centered around the true value (17\%, based on the information intervention in \cite{cyyz}). This information makes subjects update their beliefs, which yields a posterior distribution from which we calculate posterior means and modes over others' attendance. Prior and posterior means and modes in this stylized environment are summarized in Table \ref{summary}.

\begin{table}[htbp!]
\begin{center}
\caption{Summary of Mean and Modal Beliefs in the Example}
\label{summary}
\begin{tabular}{ c c c c } 
 \hline
& & Prior & Posterior\\
\hline
Group 1 ( $b^1$) & mean & 0.1645 & 0.17 \\ 
  & mode & 0.30 & 0.17 \\ 
	\hline
Group 2 ( $b^2$) & mean & 0.10 & 0.127 \\ 
& mode & 0.05 & 0.17\\
 \hline
\end{tabular}
\end{center}
\end{table}

Suppose that subjects reveal their true intentions when asked about their willingness to protest. Then, following the intervention, only Group 1 subjects would express a willingness to participate in the protest following treatment (as their posterior mean belief is above the cost, 0.165). However, the elicitation scheme in Section \ref{original_exp} satisfies Claim \ref{claim: mode}, so subjects should report their modal beliefs. This leads to the following two results (Claims \ref{claim: opp_dir} and \ref{claim: control-vs-treated}) that would seem to violate SC if one interpreted ``elicited beliefs" as mean beliefs. 

\begin{claim}
\label{claim: opp_dir}
 Despite SC, the group of subjects whose ``elicited beliefs'' went down from 0.30 to 0.17 want to
participate in the protest. The group of subjects whose ``elicited beliefs''
went up from 0.05 to 0.17, do not want to participate.
\end{claim}

Suppose that a part of Group 1 was randomly assigned to a control group (who did not receive information intervention), and the rest of Group 1 to a treatment group (who received the information intervention). 

\begin{claim}
\label{claim: control-vs-treated}
The control group would have high ``elicited prior beliefs'' of 0.30 for both posterior and prior. The treatment group's ``elicited beliefs'' would go down from 0.30 (prior) to 0.17 (posterior). Still, only the treatment group would want to participate in the protest. An observer who interprets reports as mean beliefs, would interpret this evidence to reject SC, even though SC holds by construction.
\end{claim}

By comparison, this misinterpretation would not occur if one used the elicitation scheme in Claim \ref{claim: mean}, as both groups would have reported mean beliefs, which increase following treatment. However, if no incentives were provided to subjects (as in Section \ref{sec:noinc}), then we cannot derive best responses for $r_i$ in this example.

\subsection{Means vs. Modes in \cite{cyyz}: A Quantification}\label{subsec: quantification}

While we cannot truly know whether beliefs were asymmetric and whether means and modes were updated in opposite directions (and even whether modes themselves were reported), we can illustrate that this concern may be quantitatively meaningful. To do so, we extend the previous example to a more easily quantifiable set-up, where we can provide sufficient conditions under which the mean and mode beliefs would be updated in opposite directions. 

Let the true proportion of students planning to attend the protest be $x \in [0,1]$. Assume that each subject, $i$, has beliefs over $x$ that follow a $Beta(\alpha_i,\beta_i)$ distribution, where $\alpha_i, \beta_i$ are parameters larger than 1.\footnote{\label{beta}The $Beta$ distribution has several advantages in this set-up. First, its support is bounded between $[0,1]$, just like the share of participation turnout. Second, it can be symmetric ($\alpha_i=\beta_i$) with mean and mode being equal, or asymmetric, allowing us to test our mechanism which relies on asymmetric beliefs. Third, its mode, mean and variances have closed form solutions. Fourth, for a $Beta$ distribution, condition (ii) from Claim \ref{claim: mode} holds, and hence subjects should be reporting values approximately equal to the prior and posterior modes.  Finally, it provides the convenience of being the conjugate prior when the likelihood is binomial/Bernoulli, which would be useful to model how subjects update their priors after the particular information intervention used in \cite{cyyz}. The latter is used explicitly in our proof of Claim \ref{claim: mean-mode}.}

The information intervention in \cite{cyyz} exposes subjects to the true level of planned participation of 1234 subjects, and in the experiment $\hat x=17\%$ of the 1234 subjects planned to protest. We assume that subjects interpret this information as a Binomial draw from the true, but unobserved, realized proportion $x^*$ with sample size $n = 1234$. Then:

\begin{claim}
\label{claim: mean-mode}
Any subject who started with a belief distribution such that the mean $\mu_i>0.1707$ and mode $m_i<0.1698$ (or, $m_i>0.1707$ and $\mu_i<0.1698$) would always update their mean and mode in opposite directions after observing $\hat x=0.17$.
\end{claim}

\begin{proof} In Appendix.\end{proof}

Claim \ref{claim: mean-mode} reinforces Claims \ref{claim: opp_dir} and \ref{claim: control-vs-treated}, while generalizing Figure \ref{fig1}: means and modes are often updated in opposite directions, and thus misinterpreting modes as means could reverse the relationship between participation choices and the direction of belief-update. 


We can use this set-up in a numerical exercise. Let us interpret the belief data from \cite{cyyz} as modal beliefs from a heterogenous $Beta$ distribution. As we observe different reports across individuals (which are the modes of their corresponding belief distributions), such beliefs must be heterogeneous across participants.\footnote{The only alternative would be to attribute these differences to measurement error, which would be unsatisfactory. The parametric assumptions then allow us to learn about the distributions of individual beliefs.} 

We model this heterogeneity by assuming a ``random coefficients" type structure on the parameters $\alpha_i, \beta_i$, yielding a Bayesian hierarchical model. In particular, to be consistent with the structure above, we assume $\alpha_i, \beta_i$ are i.i.d. across the population according to: $\alpha_i - 1 \sim \chi^2(\ell)$ and $\beta_i-1 \sim \chi^2(q)$, with the two distributions being independent. Next, the properties of $\chi^2$-distributions imply that the modal beliefs elicited in their experiment follow:
\begin{equation}
\frac{\alpha_i-1}{\alpha_i-1+\beta_i-1} \sim Beta\left(\frac{\ell}{2},\frac{q}{2}\right)
\label{mode_dist}
\end{equation}

Hence, the researcher with the dataset from \cite{cyyz} observes the distribution of modes across subjects, characterized by (\ref{mode_dist}), rather than individual belief distributions ($Beta(\alpha_i,\beta_i)$). Note that the former can be symmetric ($\ell=q$) even when individual beliefs are not ($\alpha_i \neq \beta_i$).\footnote{In footnote 25, \cite{cyyz} suggest mean and mode are likely to be close together. However, if they only have data on reported modes, they can only compare the mean and modes of the distribution of reported modes, rather than means and modes of the original distribution of individual beliefs.}

Identification of $\ell,q$ in (\ref{mode_dist}) follows immediately from the moments of the distribution of modes in the data. We estimate $\ell/2, q/2$ by Maximum Likelihood, yielding estimates $\hat{\ell}/2, \hat{q}/2$. Since $\ell,q$ are degrees of freedom, our estimates for them are given by their rounded up counterparts. The estimates are shown in Appendix Table \ref{tab:main}. Appendix Table \ref{tab:fit} shows that this parametrization fits the moments in the data well.

To calculate the share of beliefs with modes and means being updated in opposite directions (i.e. that satisfy the conditions outlined in Claim \ref{claim: mean-mode}, in which those two statistics are on opposite sides of 0.17), we draw $R=100000$ parameters from the joint estimated distribution of $\{(\alpha_i, \beta_i)\}_{i=1}^{n}$. We then compute the share of such draws that have the mode and mean on opposite sides of 17\%, and such that the mode $\pm 2$ percentage points has more mass than the mean $\pm 2$ percentage points. This check is needed given the authors' error tolerance.

The main result from this exercise is that, even in a simplified set-up, $34.8\%$ of drawn beliefs would have mean and modes updated in opposite directions in the original experiment where subjects would have preferred to report the mode rather than the mean. While this is merely illustrative, it does suggest that this may be a meaningful concern, and that the authors' original estimates could have had their signs reversed if they interpreted ``elicited beliefs" as modal beliefs.

\section{Concluding Remarks}

There has been increased popularity in using experiments to measure and vary individual beliefs, particularly in strategic settings. We discuss how belief elicitation schemes that look prima-facie similar, and that are widely used in experimental work, may identify different properties of the true probability distribution. In particular, we describe cases where incentive schemes rewarding subjects for guessing an outcome correctly incentivize reports of modal beliefs, mean beliefs or median beliefs. We also discuss the perils of not using an incentive scheme at all. While it is perfectly possible that, in certain settings, belief distributions are indeed symmetric and unimodal, in others they may not, which could imply that even the sign of identified effects may be switched if one interprets mean beliefs as modal beliefs and vice-versa. Hence, researchers should consider which scheme would elicit the appropriate data for their research question.

We illustrate these points using the experimental data in \cite{cyyz} within the context of political protests in Hong Kong. The authors provide an innovative way of examining how the decision to participate changes when there is an exogenous shock to citizen's beliefs. We make two main observations. First, our results imply that unless one assumes symmetric unimodal distributions, the elicitation method in \cite{cyyz} measures modal beliefs instead of the mean beliefs that the regression specification calls for. Second, we explain how interpreting modal beliefs as mean beliefs might lead researchers to erroneously reject strategic complementarity.
This suggests alternative interpretations to their findings which could be studied in future work.

\newpage
\bibliographystyle{apalike}
\bibliography{References}

\begin{thebibliography}{}

\bibitem[Barbera and Jackson, 2020]{barbera}
Barbera, S. and Jackson, M.~O. (2020).
\newblock A model of protests, revolution, and information.
\newblock {\em Quarterly Journal of Political Science}, 15(3):297--335.

\bibitem[Bueno~de Mesquita, 2010]{mesquita}
Bueno~de Mesquita, E. (2010).
\newblock Regime change and revolutionary entrepreneurs.
\newblock {\em American Political Science Review}, 104(3):446--466.

\bibitem[Bursztyn et~al., 2020]{bursztyn}
Bursztyn, L., Gonz{\'a}lez, A.~L., and Yanagizawa-Drott, D. (2020).
\newblock Misperceived social norms: Women working outside the home in saudi
  arabia.
\newblock {\em American economic review}, 110(10):2997--3029.

\bibitem[Cantoni et~al., 2019]{cyyz}
Cantoni, D., Yang, D.~Y., Yuchtman, N., and Zhang, Y.~J. (2019).
\newblock Protests as strategic games: experimental evidence from hong kong's
  antiauthoritarian movement.
\newblock {\em The Quarterly Journal of Economics}, 134(2):1021--1077.

\bibitem[Chen and Yang, 2019]{chenyang}
Chen, Y. and Yang, D.~Y. (2019).
\newblock The impact of media censorship: 1984 or brave new world?
\newblock {\em American Economic Review}, 109(6):2294--2332.

\bibitem[Chenoweth and Stephan, 2011]{chenoweth}
Chenoweth, E. and Stephan, M. (2011).
\newblock {\em Why civil resistance works: The strategic logic of nonviolent
  conflict}.
\newblock Columbia University Press.

\bibitem[Cruz et~al., 2020]{trebbi2}
Cruz, C., Keefer, P., Labonne, J., and Trebbi, F. (2020).
\newblock Making policies matter: Voter responses to campaign promises.

\bibitem[Gehlbach et~al., 2016]{gehlbach}
Gehlbach, S., Sonin, K., and Svolik, M.~W. (2016).
\newblock Formal models of nondemocratic politics.
\newblock {\em Annual Review of Political Science}, 19:565--584.

\bibitem[Hager et~al., 2022a]{hager2022group}
Hager, A., Hensel, L., Hermle, J., and Roth, C. (2022a).
\newblock Group size and protest mobilization across movements and
  countermovements.
\newblock {\em American Political Science Review}, 116(3):1051--1066.

\bibitem[Hager et~al., 2022b]{hager2022political}
Hager, A., Hensel, L., Hermle, J., and Roth, C. (2022b).
\newblock Political activists as free-riders: Evidence from a natural field
  experiment.

\bibitem[Jarke-Neuert et~al., 2021]{jarke2021free}
Jarke-Neuert, J., Perino, G., and Schwickert, H. (2021).
\newblock Free-riding for future: Field experimental evidence of strategic
  substitutability in climate protest.
\newblock {\em arXiv preprint arXiv:2112.09478}.

\bibitem[Kendall et~al., 2015]{trebbi}
Kendall, C., Nannicini, T., and Trebbi, F. (2015).
\newblock How do voters respond to information? evidence from a randomized
  campaign.
\newblock {\em American Economic Review}, 105(1):322--53.

\bibitem[Morris and Shin, 2003]{morris}
Morris, S. and Shin, H.~S. (2003).
\newblock Global games: Theory and applications.
\newblock {\em Advances in Economics and Econometrics}.

\bibitem[Olson, 2009]{olsonlogic}
Olson, M. (2009).
\newblock {\em The logic of collective action}, volume 124.
\newblock Harvard University Press.

\bibitem[Palfrey and Rosenthal, 1984]{palfrey1984}
Palfrey, T.~R. and Rosenthal, H. (1984).
\newblock Participation and the provision of discrete public goods: a strategic
  analysis.
\newblock {\em Journal of public Economics}, 24(2):171--193.

\bibitem[Passarelli and Tabellini, 2017]{passarelli}
Passarelli, F. and Tabellini, G. (2017).
\newblock Emotions and political unrest.
\newblock {\em Journal of Political Economy}, 125(3):903--946.

\bibitem[Tullock, 1971]{tullock1971}
Tullock, G. (1971).
\newblock The paradox of revolution.
\newblock {\em Public Choice}, pages 89--99.

\end{thebibliography}

\newpage

\appendix

\section{Proofs and Additional Results}\label{appA}

\subsection{Proof of Claim \ref{claim: mode}}

\begin{proof} (ii) Let $x$ denote the subject's report. The subject gets a payoff of 1 if and only if the true turnout $y \in [x-\delta, x+\delta]$. She believes $y$ is continuously distributed with density $f$.

Thus, her expected earning is $\int_{x-\delta}^{x+\delta}f(y)dy$. Now,
\begin{eqnarray*}
C&=&\frac{d}{dx}\left[\int_{x-\delta}^{x+\delta}f(y)dy\right] \\
 & = & f(x+\delta)-f(x-\delta).
\end{eqnarray*}
We are given that $f$ is increasing below $m$ and decreasing above $m$.  Thus, $x+\delta\leq m$ implies $\left(f(x+\delta)-f(x-\delta)\right)>0$ which implies $C$ is positive. Thus, at any $x+\delta\leq m$, increasing the report strictly increases the expected earning. Similarly,
$x-\delta\geq m$ implies $f(x+\delta)-f(x-\delta)<0$ which implies $C$ is negative. Thus, at any $x+\delta>m$, decreasing the report strictly increases the expected earning. Thus,
the optimal report must lie in $[m-\delta,m+\delta]$.

(iii) If $x$ is such that $(x-\delta,x+\delta)\in(m-\epsilon,m+\epsilon)$
but $x\notin[m-\delta,m+\delta]$, then reporting $x$ cannot be optimal:
this follows from part (ii). Else if, $x+\delta<m-\epsilon$ or $x-\delta>m+\epsilon$,
then, the interval $(x-\delta,x+\delta)$ must lie completely outside $(m-\epsilon,m+\epsilon)$,
and hence 
\begin{eqnarray*}
\left[\int_{x-\delta }^{x+\delta }f(y)dy\right] 
& \leq & \left[\int_{x-\delta}^{x+\delta}\max_{w\notin[m-\epsilon,m+\epsilon]}f(w)dy\right]\\
& \leq & \left[\int_{x-\delta}^{x+\delta}\min_{w\in[m-\epsilon,m+\epsilon]}f(w)dy\right]\\
& = & \left[\int_{m-\delta}^{m+\delta}\min_{w\in[m-\epsilon,m+\epsilon]}f(w)dy\right]\\
 & < & \left[\int_{m-\delta}^{m+\delta}f(y)dy\right]
\end{eqnarray*}
which implies that $x$ cannot be an optimal report. The equality in the third step follows from (i) $\min_{w\in[m-\epsilon,m+\epsilon]}f(w)$ is independent of the variable of integration and hence can be taken outside the integration sign, (ii) $(x+\delta)-(x-\delta) = (m+\delta)-(m-\delta)$ in the resulting integral.

Else if, $(x-\delta,x+\delta)$
lies partially in $(m-\epsilon,m+\epsilon)$ then it can be improved
upon similarly by reporting $m$ instead. 

Thus, the optimal report must lie in $[m-\delta,m+\delta]$.

(i) is omitted, as they are easy to show and analogous to (ii).
\end{proof}

\subsection{Proof of Claim \ref{claim: mean}}
\begin{proof}
\begin{eqnarray*}
r^{*} & = & \arg\max_{r\in[0,1]}\int_{0}^{1}\left(B-A(x-r)^{2}\right)f(x)dx\\
 & = & \arg\max_{r\in[0,1]}\left(B-A\int_{0}^{1}(x-r)^{2}f(x)dx\right)\\
 & = & \arg\min_{r\in[0,1]}\int_{0}^{1}(x-r)^{2}f(x)dx\\
 & = & \arg\min_{r\in[0,1]}\int_{0}^{1}\left((x-\bar{x})^{2}+2(x-\bar{x})(\bar{x}-r)+(\bar{x}-r)^{2}\right)f(x)dx\\
 & = & \arg\min_{r\in[0,1]}\int_{0}^{1}\left((\bar{x}-r)^{2}\right)f(x)dx\\
 & = & \bar{x}
\end{eqnarray*}
\end{proof}

\subsection{Proof of Claim \ref{claim: median}}

\begin{proof}
(i) The expected payoff from reporting $r$ is
$$\Pi(r)=\int_{0}^{1} (A-B|s-r|\Big)f(s)ds\;$$
Maximizing $\Pi(r)$ is equivalent to minimizing:
\begin{eqnarray*}
h(r)&=&\int_{0}^{1} |s-r|f(s)ds\\
        &=&\int_{0}^{r} (r-s)f(s)ds+\int_{r}^{1} (s-r)f(s)ds
\end{eqnarray*}
Therefore, 
\begin{eqnarray*}
\frac{d h(r)}{dr}&=&\int_{0}^{r} f(s)ds-\int_{r}^{1} f(s)ds\\
\end{eqnarray*}
$$h(r)\text{ is }\begin{cases}
\text{decreasing}&\text{if }r<M\\
\text{constant}&\text{if }r=M\\
\text{increasing}&\text{if }r>M\;.
\end{cases}$$

\noindent Hence, $h$ is minimized at and the expected profit is maximized at $r=M$.

(ii) Suppose that there are $n$ values in $[0,1]$, $s_1<s_2<\dots<s_n$, with $f(s_i)>0$. The expected payoff from reporting $r$ is
$$\Pi(r)=\sum_{k=1}^n\Big(A-B|s_k-r|f(s_k)\Big)\;$$
Maximizing $\Pi(r)$ is equivalent to minimizing $$h(r)=\sum_{k=1}^n(s_k-r)f(s_k)\;$$
If $r\leq s_1$, then $$h(r)=\sum_{k=1}^n|s_k-r|=\sum_{k=1}^n(s_k-r)\;$$ As $x$ increases, each term above decreases until $r$ reaches $s_1$, therefore $\Pi(s_1)>\Pi(r)$ for all $r< s_1$.

Now suppose that $s_k\le r\le r+d\le s_{k+1}$ for some $d>0$ small enough. Then,

\begin{eqnarray*}
h(r+d)&=&\sum_{i=1}^k\Big(r+d-s_i\Big)f(s_i)+\sum_{i=k+1}^n\Big(s_i-(r+d)\Big)f(s_i)\\
        &=&d\sum_{i=1}^k f(s_i)+\sum_{i=1}^k(r-s_i)f(s_i)-d\sum_{i=k+1}^n f(s_i)+\sum_{i=k+1}^n(s_i-r)f(s_i)\\
        &=&d\Big(\sum_{i=1}^k f(s_i)-\sum_{i=k+1}^n f(s_i)\Big) +h(r)
\end{eqnarray*}

Thus, $$h(r+d)-h(r)=d\Big(\sum_{i=1}^k f(s_i)-\sum_{i=k+1}^n f(s_i)\Big)$$
If $M(x)=\sum_{s_i \leq x} f(s_i)-\sum_{s_i \geq x} f(s_i)$, then,

$$h(x)\text{ is }\begin{cases}
\text{decreasing}&\text{if }M(x)<0\\
\text{constant}&\text{if }M(x)=0\\
\text{increasing}&\text{if }M(x)>0\;.
\end{cases}$$

Thus, $h(x)$ is minimal when $M(x)=0$.
\end{proof}

\subsection{Proof of Claim \ref{claim: mean-mode}}


\begin{proof} 
As $p\sim Beta(\alpha_i,\beta_i)$, prior beliefs for $i$ about the share of people attending the protests has mean $\mu_i = \frac{\alpha_i}{\alpha_i+\beta_i}$, mode $m_i = \frac{\alpha_i-1}{\alpha_i+\beta_i-2}$ and variance $\frac{\alpha_i \beta_i}{(\alpha_i+\beta_i)^2(\alpha_i+\beta_i+1)}$, which are well defined for $\alpha_i, \beta_i>1$. 

With a prior distribution $Beta(\alpha_i, \beta_i)$, the posterior distribution after observing the number of successes in the draw would be $Beta(\alpha_i',\beta_i')$ with $\alpha_i'=\alpha_i+\hat x n$ and $\beta'_i=\beta_i+(1-\hat x)n$ (neglecting rounding errors), since the Beta distribution is a conjugate prior to the binomial likelihood. Therefore, the updated mean is given by $\mu'=\frac{\alpha'}{\alpha'+\beta'}=\frac{\alpha+\hat xn}{\alpha+\beta+n}=\frac{\frac{\alpha}{n}+\hat x}{\frac{\alpha+\beta}{n}+1}\in[\frac{\hat x}{\frac{1}{n}+1},\frac{\frac{1}{n}+\hat x}{\frac{1}{n}+1}]$. 

After observing a sample of n=1234 observations, subject's posterior mean must lie within a tight interval $[0.999\hat x,0.9992\hat x+0.00081]$ centered around the observed frequency of successes $\hat x$. Similarly, one could approximate the mode to be contained in $[1.0008\hat x-0.00082,1.00082\hat x]$. Claim \ref{claim: mean-mode} follows from substituting $\hat x=0.17$.
\end{proof}

\subsection{Statement and Proof of Claim \ref{claim:beta}}

\begin{claim}
\label{claim:beta}
For a $Beta$ distribution density function $f(x;\alpha,\beta)$ with $\alpha,\beta>1$ and mode $m$, one has $f'(x)>0$ for $x<m$ and $f'(x)<0$ for $x>m$.
\end{claim}
\begin{proof}
$f(x;\alpha,\beta)=kx^{\alpha-1}(1-x)^{\beta-1}$ where $k$ is a constant. Therefore, 

\begin{eqnarray*}
f'(x)&=&k(\alpha-1)x^{\alpha-2}(1-x)^{\beta-1}-k(\beta-1)x^{\alpha-1}(1-x)^{\beta-2}\\
	&=&kx^{\alpha-2}(1-x)^{\beta-2}\left((\alpha-1)(1-x)-(\beta-1)x\right)\\
	&=&kx^{\alpha-2}(1-x)^{\beta-2}\left((\alpha-1)-(\alpha+\beta-2)x\right)\\
	&=&k(\alpha+\beta-2)x^{\alpha-2}(1-x)^{\beta-2}\left(\frac{\alpha-1}{\alpha+\beta-2}-x \right)
\end{eqnarray*}
As the mode of $f(x;\alpha,\beta)$ is $\frac{\alpha-1}{\alpha+\beta-2}$, the derivative $f'$ must be positive below the mode and negative above the mode.
\end{proof}

\newpage

\subsection{Results of the Numerical Exercise}

\FloatBarrier

\begin{table}[!htbp] \centering 
  \caption{Estimates of the parameters for the distribution of the modes, $Beta(\ell/2,q/2)$.} 
  \label{tab:main} 
\begin{tabular}{@{\extracolsep{5pt}}lccc} 
\\[-1.8ex]\hline \\[-1.8ex] 
\\[-1.8ex] &  & \multicolumn{2}{c}{Sample} \\ 
 &  & Treatment Only & Full Sample \\ 
\hline \\[-1.8ex] 
\textit{Data from First Survey} &&&\\
\\
  $\ell$/2 & &     0.232 &     0.172 \\ 
   & & [0.230, 0.235] & [0.171, 0.174] \\ 
	\\
  q/2 & & 1.383 &  1.192 \\
& & [1.333, 1.436] & [1.142, 1.245] \\ 
\hline \\[-1.8ex] 
\\
  $\mathbb{E}\alpha_i$  & & 2 & 2 \\ 
 $Var(\alpha_i)$ & & 6 & 6\\ 
\\
 $\mathbb{E}\beta_i $ & & 4 & 4 \\ 
 $Var(\beta_i)$ & & 10 & 10 \\ 
\hline \\[-1.8ex] 
\end{tabular} 
\parbox[c]{6.2in}{%
{\footnotesize{}Notes: We present the Maximum Likelihood estimates for the parameters for the distribution of the modes in the first survey, $Beta(\ell/2,q/2)$. $\ell,q$ are the hyperparameters for the distributions of individual beliefs $Beta(\alpha_i,\beta_i)$, where $\alpha_i-1 \sim \chi^2(\ell), \beta_i-1\sim \chi^2(q)$. It follows that $\mathbb{E}\alpha_i = 1+\ell$, $\mathbb{E}\beta_i = 1+q$. We present estimates using only treated subjects and using all subjects in \cite{cyyz}. Each panel shows the results from separately estimating the distributions of modes in the first and second surveys. 95\% confidence intervals in brackets (for $\ell/2,q/2$).}}
\end{table} 

As we can see from Panel A, the average $\alpha_i, \beta_i$ are clearly statistically different, so the modes and means for subjects are not the same on average.  Table \ref{tab:fit} below shows that this simple model fits key moments of the data.

\begin{table}[!htbp] \centering 
  \caption{Model Fit: How our simple model matches the distribution of the modes in the data} 
  \label{tab:fit} 
\begin{tabular}{@{\extracolsep{5pt}}lcc} 
\\[-1.8ex]\hline \\[-1.8ex] 
 & Estimated (at the average) & Data\\ 
\\[-1.8ex] Moments & & \\ 
\hline \\[-1.8ex] 
\\
    Average Mode &         0.144 & 0.140 \\ 
    Variance of Modes & 0.032 & 0.024 \\ 
\hline \\[-1.8ex] 
\hline \\[-1.8ex] \end{tabular} 
\parbox[c]{6.2in}{%
{\footnotesize{}Notes: We present the model fit from our estimates for the distribution of beliefs, which is assumed to be $Beta(\alpha_i,\beta_i)$, with $\alpha_i \sim \chi^2(p), \beta_i \sim \chi^2(q)$. We compute model fits at the average value of $\alpha_i, \beta_i$, given by the estimated counterpart to $1+p, 1+q$.}}
\end{table}  

\FloatBarrier

\clearpage

\section{Overview of the Experiment in \cite{cyyz}}\label{appB}

The experiment in \cite{cyyz} is composed by three parts. In this section, we summarize the main features that are relevant to our note.

\subsection*{Part 1: June 24, 2016 (1 week before the July 1 planned protest)}

The authors recruited participants via an email to the entire undergraduate population of HKUST on June 24, 2016. They had 1741 respondents, with 1576 of them considered as Hong Kong natives (their target population). Respondents received an average of HK\$205 (around US\$25), for completing this survey.

In this stage, the authors elicit subjects' planned participation, as well as subjects' beliefs regarding other subjects' planned protest participation, using the following two questions:

\medskip

\begin{itemize}
\item \textbf{Own Report:}\\ \textit{``Are you planning to participate in the July 1st march in 2016?"}

\medskip

The possible answers were: (i) ``Yes", (ii) ``Not sure yet, but more likely than not", (iii) ``Not sure yet, but more unlikely than yes", (iv) ``No".\\

\medskip

\item \textbf{Belief Elicitation:} \\ \textit{``Please guess what percentage of the participants from HKUST of this study plan to participate in the July 1st march in 2016 (answer either ``Yes" or ``Not sure yet, but more likely than not" to the above question on July 1st March in 2016)}

\medskip

\textit{If your guess is within 2 percentage points of the percent of students who actually answer either ``Yes" or ``Not sure yet, but more likely than not," you will earn a bonus payment of HK\$10.}
\end{itemize}

\medskip

16.9\% of subjects (whether Hong Kong natives or not) answered that they planned to attend the protest. The authors use this number, rounded up to 17\%, as the informational treatment in Part 2.

\subsection*{Part 2: June 30, 2016 (1 day before the July 1 planned protest)}

A second online survey elicited beliefs about turnout in the following day's protest through a short online survey. 1303 students responded. Students received a payment of HK\$25 for completing the survey. 

In this survey, only a random subset of subjects (the ``Treated" Group), received the following information:

\medskip

\begin{itemize}
\item
\textbf{Informational Treatment:} \\ \textit{``Recall that you guessed that [Part 1 response]\% of HKUST survey participants would plan to attend the July 1 march. Based on last week's survey, the true percentage of survey participants who plan to attend the July 1 march is 17\%."}
\end{itemize}

\medskip

However, all subjects (whether treated, who received the information above, or control, who did not) were reminded of their Part 1 responses above, and were allowed to update their answers. In particular, after the reminder, they were asked:

\medskip

\begin{itemize}
\item
\textbf{Report Updating:}
\textit{``Perhaps since then your views have changed. We now ask you again to provide guesses about actual attendance of the July 1 march. Instead of your guesses in the previous survey, we will use today's guesses to determine your bonus payment.}

\medskip

\textit{Please guess what percentage of the participants from HKUST of this study will participate in the July 1st march in 2016?
If your guess is within 2 percentage points of the percent of students who actually participate, you will earn a bonus payment of HK\$10."}
\end{itemize}



\subsection*{Part 3: July 15, 2016 (2 weeks after the July 1 planned protest)}

The last part of the experiment involved a third online survey sent via email two weeks after the protest. It was completed by 1234 Hong Kong native students (this is the final sample). Students who completed this third survey received an additional payment of HK\$25. The main question was:

\medskip

\begin{itemize}
\item
``Did you attend the July 1, 2016 March?" 
\end{itemize}

A response of ``Yes"' to this question was the authors' main outcome measure. 

\newpage

\section{A Stylized Model}\label{subsec:counter}
	
Subjects in our toy model make two choices: They choose between the decisions of attend the protest ($P_i = 1$), or not attend ($P_i = 0$). And they choose some value to report, $r_i$, about others' participation, given the researchers' payment scheme $w$. To derive how subjects would report and update their beliefs and participation decisions, we start by explicitly assuming their utility function and beliefs about $P_{-i}$. 

\medskip

\textbf{Utility:} The utility of subject $i$ from her action $P_i \in\{0, 1\}$
is given by 
\begin{equation}
\label{eq: utility}
u_i(P_i, P_{-i}, r_i)=P_i (P_{-i}-0.165) + w 1_{\{r_i \in [P_{-i}-0.02, P_{-i}+0.02]\}}
\end{equation}
where $P_{-i}$ is the \textit{realized proportion} of other subjects who also participate in the protest, and $1\{.\}$ is the indicator function. This utility function provides a simple way to model three relevant aspects. First, the utility from protesting is increasing in the number of co-protestors and, hence, we have a case of strategic complementarity. Second, based on the incentive system offered in \cite{cyyz}, subjects receive a payoff $w$ if their report, $r_i$, is within 2 percentage points of the realized $P_{-i}$. Third, that only subjects whose mean belief of $P_{-i}$ exceeds 0.165 would want to protest.

\medskip

\textbf{Prior Beliefs:} Suppose half of the population starts
with a prior that the proportion of other survey participants showing
up is 
\begin{eqnarray}
\label{eq: belief1}
b^{1}(P_{-i}) & = & \begin{cases}
0 & \text{with probability }0.30\\
0.12 & \text{with probability }0.10\\
0.17 & \text{with probability }0.15\\
0.22 & \text{with probability }0.10\\
0.30 & \text{with probability }0.35
\end{cases}
\end{eqnarray}
The other half of the population starts with a prior that the proportion
of other survey participants showing up is 
\begin{eqnarray}
\label{eq: belief2}
b^{2}(P_{-i}) & = & \begin{cases}
0.05 & \text{with probability }0.35\\
0.08 & \text{with probability }0.20\\
0.12 & \text{with probability }0.20\\
0.17 & \text{with probability }0.25
\end{cases}
\end{eqnarray}

The belief distributions are only special to the extent that they allow for high and low prior modes in the two groups respectively, which can potentially get revised in opposite directions. Also, by meeting the sufficient conditions in part (i) of Claim \ref{claim: mode} under $\delta = 0.02$, the distributions ensure that the first and second groups would report their respective prior modes $0.30$ and $0.05$ as the proportion of other survey participants showing up. Additionally, the prior means over $P_{-i}$  are 0.1645\footnote{We purposely chose the mean belief of 0.1645 and the 0.165 cut-off in the participation utility to construct our example, as they relate to empirical values in \cite{cyyz}.} and 0.10 respectively. 

As both prior means are below 0.165, none of these subjects would be willing to participate in the protest. 

\medskip

\textbf{Informational Intervention and Belief Updating:} Next, the experimenters reveal the information \textit{``Based on last week's survey, the true percentage of survey participants who plan to attend
the July 1 March is 17\%.''} We assume that the subjects interpret the information provided by the experimenters as a signal that is drawn uniformly from $[\theta-10,\theta+10]$, where $\theta$ is the true percentage
of subjects showing up for the protest.\footnote{This is how information structures are often specified in the global games literature. Such noisy interpretation is also consistent with the empirical evidence: most subjects in \cite{cyyz} update little relative to the information available from the relatively large subsample of $n>1000$ subjects, see their Figure 2 and Claim \ref{claim: mean-mode} below.}

Given our assumption on how the subjects
interpret the signal, the updated posteriors for the two groups are, respectively,
\begin{eqnarray*}
b^{1'}(P_{-i}) & = & \begin{cases}
0.12 & \text{with probability }\ensuremath{\frac{0.10}{0.35}}\\
0.17 & \text{with probability }\frac{0.15}{0.35}\\
0.22 & \text{with probability }\frac{0.10}{0.35}
\end{cases}
\end{eqnarray*}
and
\begin{eqnarray*}
b^{2'}(P_{-i}) & = & \begin{cases}
0.08 & \text{with probability }\frac{0.20}{0.65}\\
0.12 & \text{with probability }\frac{0.20}{0.65}\\
0.17 & \text{with probability }\frac{0.25}{0.65}
\end{cases}
\end{eqnarray*}

Now, when experimenters ask about the percentage of participants, subjects in both groups report 0.17, given their incentives to report posterior modes. However, the posterior means over $P_{-i}$ calculated from their posterior beliefs are 0.17 and 0.127 respectively. The prior and posterior means and modes of $P_{-i}$ are summarized in Table \ref{summary}, which allows us to derive the claims in the main text.

\end{document}